\documentclass[a4paper,11pt,english,oneside]{article}

\usepackage[english]{babel}
\usepackage[latin1]{inputenc}
\usepackage[T1]{fontenc}
\usepackage{amsmath}
\usepackage{amsfonts}
\usepackage{amssymb}
\usepackage{verbatim}
\usepackage{rotating}
\usepackage{textcomp}
\usepackage[shortlabels]{enumitem}

\usepackage{colortbl}
\usepackage{xcolor}
\usepackage{graphicx}
\usepackage{fancyhdr}
\usepackage{algorithm} 
\usepackage{algorithmic}
\usepackage{amsthm}
\usepackage{indentfirst}

\usepackage{epsfig}
\usepackage[all]{xy} 
\usepackage{listings} 
\usepackage{color}
\definecolor{mygreen}{rgb}{0,0.6,0}
\definecolor{mygray}{rgb}{0.5,0.5,0.5}
\definecolor{mymauve}{rgb}{0.58,0,0.82}
\newcommand{\then}{\implies}
\newcommand{\Cal}{\cal}

\newcommand{\A}{{\mathcal A}}
\newcommand{\B}{{\mathcal B}}

\newcommand{\M}{\mathcal{M}}
\newcommand{\NLI}{{\mathcal N}}

\newcommand{\V}{{\mathcal V}}

\newcommand{\CC}{\mathbb{C}}

\newcommand{\FF}{\mathbb{F}}
\newcommand{\KK}{\mathbb{K}}
\newcommand{\NN}{\mathbb{N}}

\newcommand{\QQ}{\mathbb{Q}}
\newcommand{\RR}{\mathbb{R}}
\newcommand{\ZZ}{\mathbb{Z}}
\newcommand{\dd}{\mathrm{d}}
\newcommand{\dist}{\mathrm{d}}

\newcommand{\N}{\mathrm{N}}

\newcommand{\NNF}{\mathrm{NNF}}

\newcommand{\Span}{\mathrm{span}}

\newcommand{\w}{\mathrm{w}}
\newcommand{\nP}{{\mathfrak n}}

\newcommand{\tP}{{\mathfrak t}}
\newcommand{\op}{{\sf p}} 
\newcommand{\Gr}{Gr\"obner}

\newcommand{\ef}{\underline{f}}
\newcommand{\eg}{\underline{g}}

\theoremstyle{plain}
 \newtheorem{thm}{Theorem}[section]
 \newtheorem{theorem}[thm]{Theorem}
 
 \newtheorem{lemma}[thm]{Lemma}

 \newtheorem{fact}[thm]{Fact}

 \newtheorem{proposition}[thm]{Proposition}
 
 \newtheorem{definition}[thm]{Definition}

 \theoremstyle{definition}
 \newtheorem{example}[thm]{Example}
\theoremstyle{remark}
 \newtheorem{remark}[thm]{Remark}

\begin{document}

\title{
An algorithmic approach using multivariate polynomials for the nonlinearity of Boolean functions
}
\author{
E.~Bellini \thanks{eemanuele.bellini@gmail.com}
\quad
T.~Mora \thanks{theomora@disi.unige.it}
\quad
M. Sala \thanks{maxsalacodes@gmail.com}\\
}
\date{}

\maketitle
\begin{abstract}
The nonlinearity of a Boolean function is a key property in deciding its suitability
for cryptographic purposes, e.g. as a combining function in stream ciphers, and so
the nonlinearity computation is an important problem for applications. 
Traditional methods to compute the nonlinearity are based on transforms, 
such as the Fast Walsh Transform.
In 2007 Simonetti proposed a method to solve the above problem seen 
as a decision problem on the existence of solutions for some
multivariate polynomial systems. 
Although novel as approach, her algorithm suffered from a direct application of \Gr\ bases 
and was thus impractical.
We now propose two more practical approaches, one that determines
the existence of solutions for Simonetti's systems in a faster way and another
that writes similar systems but over fields with a different characteristics.
For our algorithms we provide an efficient implementation in the software package MAGMA.
\end{abstract}

\begin{center}
\end{center}

\section*{Introduction}
Any function from $(\FF_2)^n$ to $\FF_2$ is called a Boolean function. Boolean functions are important in symmetric cryptography, since they are used in the confusion layer of ciphers. 
An affine Boolean function does not provide an effective confusion. To overcome this, we need functions which are as far as possible from being an affine function. 
The effectiveness of these functions is measured by several parameters, one of these is called nonlinearity (\cite{CGC-cd-book-carlet}).
Traditional methods to
compute the nonlinearity are based on transforms, such as the Fast Walsh Transform.
In \cite{CGC-cd-art-ilawcc07} a method was proposed by Simonetti that was based on an interpretation of the
above problem as a decision problem on the existence of solutions for some
multivariate polynomial systems with Boolean variables. 
This was the first time that a classical problem in Boolean functions was tackled with multivariate polynomial techniques.
Although novel and interesting as approach, her algorithm suffered from a direct application of \Gr\ bases and was thus impractical.
In this paper we propose two more practical approaches, one that determines
the existence of solutions for Simonetti's systems in an efficient way and another
that writes similar systems but over fields with a different characteristics.\\
For our algorithms we provide an efficient implementation in the software package MAGMA,
which was presented at the conference MEGA 2015 \cite{CGC-cry-talk-bellini2015}.\\
Although the complexity of our method is still far from the best-known methods using Fast Walsh Transforms or similar techniques, 
we believe that our improvement on Simonetti's original idea is significant and that  there is still space for improvement on multivariate-polynomials methods to solve this problem, which might also lead to new insights.\\
The structure of this paper is as follows.
In Sections \ref{prel} we recall our preliminaries, especially regarding Boolean functions and related polynomials.
In Section \ref{secNLwithGBoverF2} we describe Simonetti's approach.
In Section \ref{traverso} we sketch a strategy to solve Simonetti's
systems which does not require the computation of a Groebner basis,
thanks to a 1992 idea by Traverso,
and we thus describe our first algorithm, Algorithm 1, which is a refined version of Simonetti's algorithm.
In Section  \ref{secNLwithGBoverQ} we associate to each Boolean function 
in $n$ variables a rational polynomial whose evaluations represent the distance from all possible affine functions,
obtaining our second algorithm, Algorithm 2.
In Section \ref{secNLwithFPE}, we use an alternative approach to compute the nonlinearity 
avoiding the hard task of solving a polynomial system of equations, which is described in Algorithm \ref{algNLfromNLP}.\\ 
In Section \ref{secNLComplexity} we provide some complexity considerations.
\section{Preliminaries}
\label{prel}
\subsection{Nonlinearity of Boolean functions}
  \label{secPrelOnBF}

In this section we summarize some definitions and known results from \cite{CGC-cd-book-carlet} and \cite{CGC-cd-book-macwilliamsI}, concerning Boolean functions and the classical techniques to determine their nonlinearity.\\

We denote by $\FF$ the field $\FF_2$. The set $\FF^n$ is the set of all binary vectors of length $n$, viewed as an $\FF$-vector space.\\
Let $v\in\FF^n$. The \emph{Hamming weight} $\w(v)$ of the vector $v$ is the number of its nonzero coordinates. For any two vectors $v_1,v_2\in\FF^n$, the \emph{Hamming distance} between $v_1$ and $v_2$, denoted by $\dist(v_1,v_2)$, is the number of coordinates in which the two vectors differ.\\
A \emph{Boolean function} is a function $f:\FF^n\rightarrow \FF$. The set of all Boolean functions from $\FF^n$ to $\FF$ will be denoted by $\B_n$.\\
We assume implicitly to have ordered $\FF^n$, so that $\FF^n=\{\op_1,\ldots,\op_{2^n}\}$.\\ 
A Boolean function $f$ can be specified by a \emph{truth table}, which gives the evaluation of $f$ at all $\op_i$'s.
\begin{definition}
We consider the evaluation map:
$$
\B_n \longrightarrow \FF^{2^n} 
\qquad
f \longmapsto \underline{f}=(f(\op_1),\ldots,f(\op_{2^n}))\,.
$$
The vector $\underline{f}$ is called the \emph{evaluation vector} of $f$.
\end{definition}
Once the order on $\FF^n$ is chosen, i.e. the $\op_i$'s are fixed, it is clear that the evaluation vector of $f$ uniquely identifies $f$.\\
A Boolean function $f\in\B_n$ can be expressed in a unique way as a square free polynomial in $\FF[X]=\FF[x_1,\ldots,x_n]$, i.e.
$$f=\sum_{v \in \FF^n}b_vX^v\,,$$
where $X^v=x^{v_1}\cdots x^{v_n}$.\\
This representation is called the \emph{Algebraic Normal Form} (ANF).\\
\begin{definition}
The degree of the ANF of a Boolean function $f$ is called the \emph{algebraic degree} of f, denoted by $\deg f$, and it is equal to \\
$\max\{\w(v) \mid v \in \FF^n, b_v \ne 0 \}$.
\end{definition}
Let $\A_n$ be the set of all affine functions from $\FF^n$ to $\FF$, i.e. the set of all Boolean functions in $\B_n$ with algebraic degree 0 or 1. If $\alpha\in\A_n$ then its ANF can be written as
$$\alpha(X)=a_0 + \sum_{i=1}^na_ix_i\,.$$
\indent
There exists a well known divide-and-conquer butterfly algorithm (see \cite{CGC-cd-book-carlet}, p.10) to compute the ANF from the truth-table (or vice-versa) of a Boolean function, which requires ${\Cal O}(n2^{n})$ bit sums, while ${\Cal O}(2^n)$ bits must be stored. This algorithm is known as the \emph{fast M\"obius transform}.\\
%
In \cite{CGC-cry-art-carlet1999} a useful representation of Boolean functions for characterizing several cryptographic criteria (see also \cite{CGC-cry-art-carlet2001bent}, \cite{CGC-cry-carlet2002coset}) is introduced.\\
Boolean functions can be represented as elements of $\KK[X]/\langle X^2-X \rangle$, where $\langle X^2-X \rangle$ is the ideal generated by the polynomials $x_1^2-x_1,\ldots,x_n^2-x_n$, and $\KK$ is $\ZZ$, $\QQ$, $\RR$, or $\CC$.
\begin{definition}\label{defNNF}
 Let $f$ be a function on $\FF^n$ taking values in a field $\KK$. We call the \emph{numerical normal form (NNF)} of $f$ the following expression of $f$ as a polynomial:
 $$
 f(x_1,\ldots,x_n) = \sum_{u \in \FF^n}\lambda_u (\prod_{i=1}^{n}x_i^{u_i}) = \sum_{u \in \FF^n}\lambda_{u}X^u\,,
 $$
 with $\lambda_{u} \in \KK$ and $u=(u_1,\ldots,u_n)$.
\end{definition}
It can be proved 
that any Boolean function $f$ admits a unique numerical normal form.
As for the ANF, it is possible to compute the NNF of a Boolean function from its truth table by means of an algorithm similar to a fast Fourier transform, thus requiring ${\Cal O}(n2^n)$ additions over $\KK$ and storing ${\Cal O}(2^n)$ elements of $\KK$.\\
\indent
From now on let $\KK = \QQ$.\\
The truth table of $f$ can be recovered from its NNF by the formula $$f(u)=\sum_{a\preceq u}\lambda_a,\forall u \in \FF^n\,,$$
where $a\preceq u\iff \forall i \in \{1,\ldots,n\} \; a_i \le u_i$. Conversely, 
it is possible to derive an explicit formula for the coefficients of the NNF by means of the truth table of $f$.
\begin{proposition}\label{propNNFcoeff}
 Let $f$ be any integer-valued function on $\FF^n$. For every $u\in \FF^n$, the coefficient $\lambda_u$ of the monomial $X^u$ in the NNF of $f$ is:
 \begin{equation}\label{eqNNFCoeff}
  \lambda_u = (-1)^{\w(u)}\sum_{a\in \FF^n |a\preceq u}(-1)^{\w(a)}f(a)\,.
 \end{equation}
\end{proposition}
\begin{definition}
Let $f,g\in\B_n$. The distance $\dist(f,g)$ between $f$ and $g$ is the number of $v\in\FF^n$ such that $f(v)\neq g(v)$.
\end{definition}
The following lemma is obvious:
\begin{lemma}\label{distance}
Let $f,g$ be two Boolean functions. Then
$$\dist(f,g)=\dist(\ef,\eg)=\w(\ef+\eg)\,.$$
\end{lemma}
\begin{definition}
Let $f\in\B_n$. The \emph{nonlinearity} of $f$ is the minimum of the distances between $f$ and any affine function
$$\N(f)=\min_{\alpha\in\A_n}\dist(f,\alpha)\,.$$
\end{definition}
The maximum nonlinearity for a Boolean function $f$ is bounded by:
\begin{align}\label{eqMaxNL}
\max\{\N(f) \mid f\in\B_n\} \le 2^{n-1}-2^{\frac{n}{2}-1}\,.
\end{align}
\begin{definition}
The \emph{Walsh transform} of a Boolean function $f\in\B_n$ is the following function:
$$
\hat{F}: \FF^n \longrightarrow \mathbb{Z} 
\qquad
x \longmapsto \sum_{y\in\FF^n}(-1)^{x\cdot y + f(y)}\,.
$$
where $x\cdot y$ is the scalar product of $x$ and $y$.
\end{definition}
We have the following fact:
\begin{fact}
 $$\N(f)=\min_{v\in\FF^n}\{2^{n-1}-\frac{1}{2}\hat{F}(v)\}=2^{n-1}-\frac{1}{2}\max_{v\in\FF^n}\{\hat{F}(v)\}$$
\end{fact}
\begin{definition}
  The set of integers $\{\hat{F}(v) \mid v\in\FF^n\}$ is called the \emph{Walsh spectrum} of the Boolean function $f$.
\end{definition}
It is possible 
to compute the Walsh spectrum of $f$ from its evaluation vector in ${\Cal O}(n2^{n})$ integer operations, while storing ${\Cal O}(2^n)$ integers, by means of the \emph{fast Walsh transform} (the Walsh transform is the Fourier transform of the sign function of $f$).
Thus the computation of the nonlinearity of a Boolean function $f$, when this is given either in its ANF or in its evaluation vector, requires ${\Cal O}(n2^n)$ integer operations and a memory of ${\Cal O}(2^n)$.\\
\indent
Faster methods are known in particular cases, for example when the ANF is a sparse polynomial \cite{CGC-cry-phdthesis-calik2013}, \cite{CGC-cry-art-calik2013nonlinearity}.
%
%
\subsection{Polynomials and vector weights}
  \label{secPrelOnGB}
Here we present some results from 
\cite{CGC-cd-art-ilawcc07}, \cite{CGC-cd-inbook-D1simonetti}, \cite{CGC-cd-art-elemanumax}.
Let $\KK$ be a field and $X=\{x_1,\ldots,x_s\}$ be a set of variables. We denote by $\KK[X]$ the multivariate polynomial ring in the variables X. If $f_1,\ldots,f_N \in \KK[X]$, we denote by $\langle\{f_1,\ldots,f_N\}\rangle$ the ideal in $\KK[X]$ generated by $f_1,\ldots,f_N$. 
Let $I$ be an ideal in $\KK[X]$, we denote by $\V(I)$ its variety, that is the set of its zeros in the algebraic closure of $\KK$. \\
Let $q$ be the power of a prime.
We denote by $E_q[X]=\{x_1^q-x_1,\ldots,x_s^q-x_s\}\,,$ the set of field equations in $\FF_q[X]=\FF_q[x_1,\ldots,x_s]$, where $s\geq 1$ is an integer, understood from now on. We write $E[X]$ when $q=2$.
\begin{definition} 
Let $1\leq t \leq s$ and ${\mathsf m}\in \FF_q [X]$. We say that ${\mathsf m}$ is 
a {\bf{square free monomial}} of degree $t$ (or a {\bf{simple $t$-monomial}}) if: 
$$
{\mathsf m} = x_{h_1}\cdots x_{h_t}, \textrm{ where } h_1,\ldots, h_t \in \{1,\ldots,s\} \textrm{ and } h_\ell \neq h_j, \forall \ell\neq j\, ,$$ 
i.e. a monomial in $\FF_q [X]$ such that $\deg_{x_{h_i}}({\mathsf m})=1$ for any $1\leq i \leq t$, and $0$ otherwise. 
We denote by $\mathcal{M}_{s,t}$  the set of all square free monomials of degree $t$ in $\FF_q [X]$.
\end{definition}
Let $t\in\NN$, with $1\leq t\leq s$ and let $I_{s,t}\subset\FF_q[X]$  be the following ideal
$$I_{s,t}=\langle\{\sigma_t,\ldots,\sigma_s\}\cup E_q[X]\rangle\,,$$
where $\sigma_i$ are the elementary symmetric functions:
$$
\begin{array}{lcl}
\sigma_1 & = & x_1+x_2+\cdots+x_s,\\
\sigma_2 & = & x_1x_2+x_1x_3+\cdots+x_1x_s+x_2x_3+\cdots+x_{s-1}x_s,\\
 & \cdots\\
 \sigma_{s-1} & = & x_1x_2x_3\cdots x_{s-2}x_{s-1}+\cdots+x_2x_3\cdots x_{s-1}y_s,\\
 \sigma_s & = & x_1x_2\cdots x_{s-1}x_s.
\end{array}
$$
We also denote by $I_{s,s+1}$ the ideal $\langle E_q[X] \rangle$.
For any $1\leq i\leq s$, let $P_i$ be the set which contains all vectors in $(\FF_q)^n$ of weight $i$, $P_i=\{v\in\FF_q^n\mid \w(v)=i\}$, and let $Q_i$ be the set which contains all vectors of weight up to $i$, $Q_i=\sqcup_{0\leq j\leq i}P_j$ .
\begin{theorem}\label{gpeso}
Let $t$ be an integer such that $1\leq t\leq s$. Then the vanishing ideal $\mathcal{I}(Q_t)$ of $Q_{t}$ is
$$\mathcal{I}(Q_t)=I_{s,t+1}\,,$$
and its reduced \Gr\ basis $G$ is
$$
\begin{array}{lcl}
G=E_q[X]\cup\mathcal{M}_{s,t}\,, & \quad & \textrm{for } t\geq 2\,,\\
G=\{x_1,\ldots,x_s\}\,,& \quad & \textrm{for } t=1\,. 
\end{array}
$$
\end{theorem}
%
%
Let $\FF_q[Z]$ be a polynomial ring over $\FF_q$. Let ${\mathsf m}\in\mathcal{M}_{s,t}$, ${\mathsf m}=z_{h_1}\cdots z_{h_t}$. For any polynomial vector $W$ in the module $(\FF_q[Z])^n$, $W=(W_1,\ldots,W_n)$, we denote by ${\mathsf m}(W)$ the following polynomial in $\FF_q[Z]$:
$${\mathsf m}(W)=W_{h_1}\cdot\ldots\cdot W_{h_t}\,.$$
%
%
\subsection{A method for the quotient algebra of zero-dimensional ideals}
\label{intro_traverso}
We briefly recall the notions of 
Gr\"obner description, natural representation (or Gr\"obner representation) and linear representation,
which can be found in \cite{CGC-alg-book-teo2}.\\
Let $X = x_1,\ldots,x_k$.
Let $\M$ be the set of monomials in $\KK[X]$.
Let $J \subset {\KK[X]}$ be a zero-dimensional ideal,
$\deg(J) = s,$ and denote ${\sf A} := {\KK[X]}/J$ the corresponding
quotient
algebra, which satisfies $\dim_{\KK}({\sf A}) = s.$

For any $f\in{\KK[X]}$, we will denote $[f]\in{\sf A}$ its residue class
modulo $J$ and  $\Phi_f$ the endomorphism
$\Phi_f : {\sf A} \to {\sf A}$ defined by
$$\Phi_f([g]) = [fg] \forall [g]\in{\sf A}.$$

If we fix any ${\KK}$-basis
${\bf b} = \{[b_1],\ldots,[b_s]\}$ of ${\sf A}$ so that
${\sf A} = \Span_{\KK}({\bf b}),$
then for each $g\in{\KK[X]},$ there is a unique (row) vector, the {\em Gr\"obner
description of $g$},
$${\bf Rep}(g,{\bf b}) := \left(\gamma(g,b_1,{\bf b}),\ldots,\gamma(g,b_s,{\bf b})\right)\in
{\KK}^s$$ which satisfies
$$[g] = \sum_j \gamma(g,b_j,{\bf b}) [b_j]$$
and the endomorphism
$\Phi_f$ is naturally represented by the square matrix
$$M([f],{\bf b}) = \bigl(\gamma(fb_i,b_j,{\bf b})\bigr)
: \Phi_f(b_i) = [fb_i] = \sum_j \gamma(fb_i,b_j,{\bf b}) [b_j].$$

\begin{definition}
A  {\em natural representation}
 of $J$ is the assignement of
\begin{itemize}
\item  a ${\KK}$-basis  ${\bf b} = \{[b_1],\ldots,[b_s]\}\subset{\sf A}$
 and
\item the square matrices $A_h := \left(a_{ij}^{(h)}\right) = M([x_h],{\bf b})$
for each $h, 1 \leq h \leq k$.
 \end{itemize}
 \end{definition}

Remark that, for each $f(x_1,\ldots,x_k)\in{\KK[X]}$,
$M([f],{\bf b}) = f(A_1,\ldots,A_k)$.\\
An equivalent (via the remark above) definition of natural representation can
require the further assignement of
\begin{itemize}
\item $s^3$ values $\gamma_{ij}^{(l)}\in {\KK}$ such that
$$[b_i b_j] = \sum_l \gamma_{ij}^{(l)} [b_l]$$
for each $i, j, l, 1 \leq i,j,l\leq s.$
\end{itemize}
This notion was introduced in \cite{CGC-alg-misc-traverso92informalmega,CGC-alg-misc-traverso92draft} and reconsidered in \cite{CGC-alg-art-duality1}, \cite[Definition~29.3.3]{CGC-alg-book-teo2} under the name of {\em Gr\"obner representation}. \\
The endomorphism $\Phi_f$ and its represention $M([f],{\bf b})$ were introduced, with $f$ a linear form, in \cite{CGC-alg-art-auzinger88} as a tool for efficient solving 0-dimensional ideals.\\
If $J$ is given by its Gr\"obner basis wrt a term-ordering $<$ its natural
(actually: ``linear'' with the definition below) representation can be obtained
via \cite[Procedure~3.1]{CGC-alg-art-fglm}.\\
If  $J$ is an affine complete intersection defined by $r$ polynomials a
natural representation of it can be efficiently computed via Cardinal-Mourren
Algorithm \cite{CGC-alg-phdthesis-Cardinal93,CGC-alg-art-mourrain05}.
%
Recalling that a set ${\mathrm N}\subset{\mathcal M}$ is called an {\em escalier} if it is an {\em order ideal}, 
{\em i.e.} if for each $\lambda,\tau\in{\mathcal M}$, $\lambda\tau\in {\mathrm N} \then \tau\in {\mathrm N}$ and properly extending \cite[Definition~29.3.3]{CGC-alg-book-teo2} we set
\begin{definition}
 A natural representation is called a {\em linear representation}
 iff the basis 
 ${\bf b}$ of the representation is an {\em escalier}.
\end{definition}
%

%
Traverso introduced an algorithm in a scenario related to Gr\"obner
bases computation of a
zero-dimensional ideal $I$ (informal talk at MEGA 1992).\\
%
The setting was reformulated in \cite{CGC-alg-book-teo2}, Algorithm~29.3.8, as
follows:
given a zero-dimensional ideal $I\subset{\FF_q[X]}$ via its
natural representation
$${\bf b} = \{b_1,\ldots,b_s\}, b_1 = 1, M :=
\Bigl\{\left(a_{lj}^{(h)}\right), 1\leq h \leq r\Bigr\},$$
and a  finite set of elements $F := \{g_1,\ldots, g_t\}\subset{\FF_q[X]}$, given
via
their Gr\"obner descriptions
$${\mathsf c}^{(i)} = (c^{(i)}_1,\ldots,c^{(i)}_s),
c^{(i)}_j = \gamma(g_i,b_j,{\bf b}) \forall i, j,1\leq i \leq t, 1\leq j \leq
s,$$
so that $g_i - \sum_{j=1}^s c^{(i)}_j b_j \in I,$ for each $i$,
compute
with good complexity the linear representation of the ideal
$J := I\cup{\mathbb I}(F).$

The basic idea of the algorithm is the following:
if we consider an element
$g\in F$, having the Gr\"obner description $$g - \sum_{j=1}^\iota c_j b_j\in{\mathsf
I},
\quad c_\iota\neq 0,$$
and we enlarge $I$ by
adding $g$ to it, then we obtain the relation
$$b_\iota \equiv  -\sum_{j=1}^{\iota-1}  c_\iota^{-1} c_j b_j \bmod{{\mathsf
I}\cup\{g\}};$$
the decomposition
${\FF_q[X]} = I\oplus\Span_{\FF_q}({\bf b})$ of ${\FF_q[X]}$ into disjoint
 ${\FF_q}$-vectorspaces is then transformed into
$${\FF_q[X]} = \left(I\cup\{g\}\right) \oplus\Span_{\FF_q}({\bf b}\setminus\{b_\iota\}),$$
and we only have to substitute, in
each Gr\"obner description $\sum_{j=1}^s d_j b_j$ of the polynomials $g_i$ and
$X_h b_l$ --- which are
respectively encoded in the vectors ${\mathsf c}^{(i)}$ and in the rows
$\left(a_{l1}^{(h)},\ldots,a_{ls}^{(h)}\right)$
of the matrices of $M$ ---
the instances of $b_\iota$ with  $-\sum_{j=1}^{\iota-1}  c_\iota^{-1} c_j b_j$
thus getting
$\sum_j (d_j - c_\iota^{-1} c_j d_\iota) b_j$.

Since $J$ is an ideal, the inclusion in it of $g$ implies that $J$
necessarily contains also
the polynomials $X_h g$; note that, if the current natural representation is
$$({\bf b}', M') :
{\bf b}' := \{b'_1,\ldots,b'_\sigma\}, M' = M({\bf b}') :=
\Bigl\{\left(d_{lj}^{(h)}\right)\Bigr\}$$ and $g = \sum_{l=1}^s c_l b'_l$
then
$$X_h g =   \sum_{l=1}^s c_l X_h b'_l =
\sum_{j=1}^s \left(\sum_{l=1}^s c_l d_{lj}^{(h)}\right) b_j$$
which must be inserted in the list $F$
in order to be treated in the same way.

At termination, if $H\subset\{1,\ldots,n\}$ denotes the set of indices of the
elements
$b_j$ which have not being removed from ${\bf b}$ in this procedure, then
$J$ is described by the natural representation
$${\bf b}' = \{b_j, i\in H\}, M' =
\{\left(a_{lj}^{(h)}\right), l,j\in H, 1\leq h \leq n\}.$$

We observe that Traverso's Algorithm needs to perform at most $s$ {\bf While}-loops,
each costing ${\Cal O}(ns^2)$
\section{Simonetti's polynomial systems for the nonlinearity}
  \label{secNLwithGBoverF2}

In this section we want to tackle the following problem: 
to find a method to compute the nonlinearity of a given Boolean function $f \in \B_n$
by constructing a finite number of polynomial systems over $\FF_2$ with $N$ variables and such that:
\begin{enumerate}
 \item [A)] $N$ is of the order of $n$,
 \item [B)] the nonlinearity is obtained by merely deciding which of these systems have a binary solution.
\end{enumerate}
Since the maximum nonlinearity is of the order of $2^{n-1}$, 
we are satisfied if the number of systems we have to construct does not exceed $2^{n-1}$.\\

  In this section we 
  report the solution of the above problem, given by Simonetti in \cite{CGC-cd-art-ilawcc07}, which depends on 
  Theorem \ref{gpeso}.\\
  The starting idea is to define an ideal such that a point in its variety corresponds to an affine function with distance at most $t-1$ from $f$.\\

%

Let $A$ be the variable set $A=\{a_i\}_{0\leq i\leq n}$. We denote by $\mathfrak{g}_n\in\FF[A,X]$
the following polynomial:
$$\mathfrak{g}_n=a_0+\sum_{i=1}^n a_i x_i\,
\,.$$
%
%
\noindent According to Lemma \ref{distance}, determining the nonlinearity of $f\in\B_n$ is the same as finding the minimum weight of the vectors in the set $\{\ef+\eg\mid g\in\A_n\}\subset\FF^{2^n}$.
%
We can consider the evaluation vector of the polynomial $\mathfrak{g}_n$ as follows:
$$\underline{\mathfrak{g_n}}=(\mathfrak{g}_n(A,\op_1),\ldots,\mathfrak{g}_n(A,\op_{2^n}))\in (\FF[A])^{2^n}\,.$$
\begin{definition}\label{defIdealF2}
We denote by $J_t^n(f)$ the ideal in $\FF[A]$:
\begin{align*}
J_t^n(f) 
=
\langle
& 
\{
{\mathsf m} \big (\mathfrak{g}_n(A,\op_1)+ f(\op_1),\ldots,\mathfrak{g}_n(A,\op_{2^n})+f(\op_{2^n})\big) \mid {\mathsf m}\in\mathcal{M}_{2^n,t}
\}
\cup 
E[A]
\rangle\\
= 
\langle
&
\{{\mathsf m}(\underline{\mathfrak{g}_n}+\ef)\mid {\mathsf m}\in\mathcal{M}_{2^n,t}\}\cup E[A]\rangle\,.
\end{align*}
\end{definition} 
\begin{remark}
As $E[A]\subset J_t^n(f)$, $J_t^n(f)$ is zero-dimensional and radical.
\end{remark}
\begin{lemma}\label{nf}
For $1\leq t\leq 2^n$ the following statements are equivalent:
\begin{enumerate}
\item $\mathcal{V}(J_t^n(f))\neq \emptyset$,
\item $\exists u\in \{\ef+\eg\mid g\in\A_n\} \textrm{ such that } \w(u)\leq t-1$,
\item $\exists \alpha\in\A_n \textrm{ such that } \dist(f,\alpha)\leq t-1$.
\end{enumerate}
\end{lemma}
%
%
From Lemma \ref{nf} we immediately have the following theorem.
\begin{theorem}\label{Nf}
Let $f\in\B_n$. The nonlinearity $\N(f)$ is the minimum $t$ such that $\mathcal{V}(J_{t+1}^n(f))\neq \emptyset$.
\end{theorem}
From this theorem we can derive an algorithm to compute the nonlinearity for a function $f\in\B_n$, by 
determining if the variety of the ideal $J_t^n(f)$ has a solution or not.
%
\begin{algorithm}[H]
\caption{Basic algorithm to compute the nonlinearity of a Boolean function by finding if a solution of a polynomial systems over $\FF$ exists}
\label{algNLoverF2}
  \begin{algorithmic}[1]
    \REQUIRE{a Boolean function $f$}
    \ENSURE{the nonlinearity of $f$}
    \STATE{$j \leftarrow 1$}
    \WHILE{${\mathcal V}(J_j^n(f)) =\emptyset $} 
      \STATE{$j \leftarrow j+1$} 
    \ENDWHILE
    \RETURN $j-1$
  \end{algorithmic}
\end{algorithm}
Simonetti's systems $J_j^n(f)$ are the solutions of the problem we stated at the beginning of this section: 
they use only $n+1$ variables and all we want to know from them (in the worst case) is whether they have a solution or not. 
Observe also that the solution we are interested in does not lie in some extension field but it must remain in $(\FF_2)^{n+1}$.\\
Moreover, the number of systems we need to check is, in the worst case, the maximum nonlinearity plus one.
We claim that with our constraints Simonetti's solution is, in principle, still the best-known.\\
However, a practical application of Algorithm \ref{algNLoverF2} was missing in Simonetti's work, where she would use straightforward applications of \Gr\ bases.
\begin{remark}
If $f$ is not affine, we can start our check from $J_2^n(f)$.
\end{remark}
\section{A faster algorithm for solving Simonetti's systems}
\label{traverso}
As we have seen in Section \ref{secNLwithGBoverF2}, the nonlinearity of a Boolean function can be computed 
solving polynomial systems
over $\FF$. It is sufficient to find the minimum $j$ such that the variety of the ideal $J_t^n(f)$ is not empty. Recall that
$$
J_t^n(f) = \langle\{{\mathsf m}(\underline{\mathfrak{g}_n}+\ef)\mid {\mathsf m}\in\mathcal{M}_{2^n,t}\}\cup E[A]\rangle\,.
$$
This method becomes impractical even for small values of $n$, since $\binom{2^n}{t}$ monomials have to be evaluated. 
A first slight improvement could be achieved by adding to the ideal one monomial evaluation at a time and check if 1 has appeared in the \Gr\ basis. 
Even this way, the algorithm remains very slow.\\
Of course, an actual implementation would take care to reduce modulo the temporary basis any monomial before adding it to the computation, but it would still remain too slow.
To achieve a real improvement we need to use Traverso's strategy explained in Section \ref{intro_traverso}.
In particular, our proposal is to start with a trivial monomial basis given by the all monomials in $N(E(X))$ and then
adding a monomial ${\mathsf m}(\underline{\mathfrak{g}_n}+\ef)$ at a time, computing the new Hilbert staircase and the associated algebra-multiplication matrix. 
This way we will not get at the end a \Gr\ basis, but we would see from the final Hilbert staircase whether the ideal is trivial or not.
Since only linear algebra operations are required, this method is much faster and probably is the fastest that can be used to solve Simonetti's systems.
However, in the next section we will propose an even better method, by modifying Simonetti's systems to other fields.
\begin{theorem}
 Solving the reduced Simonetti's systems $J_j^n(f)$ using Traverso's algorithm requires ${\Cal O}(n2^{2n})$ elementary operations.
\end{theorem}
\begin{proof}
 As noted at the end of Section \ref{intro_traverso}, Traverso's algorithm requires ${\Cal O}(rs^2)$ elementary operations, 
 where $r$ is the number of variables of the equations in the system
 and $s$ is the number of monomials in the escalier. 
 In our case $r={n+1}$ and $s \le (2^{n+1})^2$.
\end{proof}
Actually the complexity exposed in the previous theorem is a large upper bound of the real complexity, 
since the monomials in the escalier could be much less than $(2^{n+1})^2$, 
though a precise estimation of their exact number is not known to the authors.\\
Since the nonlinearity of a Boolean function is bounded by $2^{n-1}-2^{n/2-1}$, Algorithm 1 implies we have to solve at most $2^{n-1}$ Simonetti's systems.\\
On the other hand only the last system dominates the computation since the previous ones do not have a solution.\\
If we suppose such systems were somehow given for free, 
computing the nonlinearity of the Boolean function with Algorithm \ref{algNLoverF2} 
would require ${\Cal O}(n2^{3n})$ operations.

%
\section{Nonlinearity and polynomial systems over $\QQ$}
  \label{secNLwithGBoverQ}
%
%
Here we present an algorithm to compute the nonlinearity of a Boolean function 
by solving a polynomial system of equations
over $\QQ$ rather than over $\FF$, which turns out to be
much faster than Algorithm \ref{algNLoverF2}. 
The same algorithm can be slightly modified to work over the field $\FF_p$, where $p$ is a prime. 
The complexity of these algorithms will be analyzed in Section \ref{secNLComplexity}.\\
%
%
%
For each $i=1,\ldots,2^n$, let us denote:
$$f_{i}^{(\FF)}(A)=\mathfrak{g}_n(A,\op_i)+f(\op_i)$$
the Boolean function where as usual $A = \{a_0,\dots,a_n\}$ are the $n+1$ variables representing the coefficient of a generic affine function.\\
In this case we have that:
$$(f_1^{(\FF)}(A),\dots,f_{2^n}^{(\FF)}(A)) = \underline{\mathfrak{g}_n}(A)+\ef \in (\FF[A])^{2^n}$$
Note that the polynomials $f_{i}^{(\FF)}$ are affine polynomials. 
\\
We also denote by
$$f_i^{(\ZZ)}(A) = \NNF(f_{i}^{(\FF)}(A))$$
the NNF of each $f_{i}^{(\FF)}(A)$ (obtained as in \cite{CGC-cry-art-carlet1999}, Theorem 1).
\begin{definition}\label{defNLP}
 We call $\nP_f(A) = f_1^{(\ZZ)}(A)+\dots+f_{2^n}^{(\ZZ)}(A) \in \ZZ[A]$ the {\bf integer nonlinearity polynomial} (or simply the \emph{nonlinearity polynomial}) of the Boolean function $f$.\\
 For any $t\in \NN$ we define the ideal $\NLI_f^t \subseteq \QQ[A]$ as follows:
 \begin{align}
  \NLI_f^t 
  & =
  \langle E[A] \bigcup \{ f_1^{(\ZZ)}+\dots+f_{2^n}^{(\ZZ)}-t \} \rangle = \\
  & = 
  \langle E[A] \bigcup \{ \nP_f-t \} \rangle
 \end{align}
\end{definition}
 Note that the evaluation vector $\underline{\nP_f}$ represents all the distances of $f$ from all possible affine functions (in $n$ variables).
\begin{theorem}
The variety of the ideal $\NLI_f^t$ is non-empty if and only if
the Boolean function $f$ has distance $t$ from an affine function. 
In particular, $\N(f) = t$, where $t$ is the minimum positive integer such that $\mathcal{V}(\NLI_f^t)\ne \emptyset$.
\end{theorem}
\begin{proof}
 Note that 
 $$\NLI_f^t = \langle E[A] \rangle + \langle \{ \nP_f(A)-t \} \rangle$$
 and so 
 $$\mathcal{V}(\NLI_f^t) = \mathcal{V}(\langle E[A]\rangle) \cap \mathcal{V}(\langle \{ \nP_f(A)-t \} \rangle)\,.$$
 Therefore $\mathcal{V}(\NLI_f^t) \ne \emptyset$ if and only if 
 $\exists \bar{a}=(\bar{a}_0, \ldots, \bar{a}_n) \in \mathcal{V}(\langle E[A]\rangle)$ such that $\nP_f(\bar{a})=t$.\\
 Let $\alpha \in \mathcal{A}_n$ such that $\alpha(X) = \bar{a}_0 + \sum_{i=1}^n \bar{a}_ix_i$.\\
 By definition we have 
 $$f_i^{(\ZZ)} = 1 \iff f(\op_i) \ne \alpha(\op_i)$$ 
 and 
 $$f_i^{(\ZZ)} = 0 \iff f(\op_i) = \alpha(\op_i)\,.$$
 Hence 
 \begin{align*}
  \nP_f(\bar{a}) = \sum_{i=1}^{2^n}f_i^{(\ZZ)}(\bar{a})-t = 0 & \iff |\{i \mid f(\op_i)\ne \alpha(\op_i) \}|=t \\
                                                              & \iff \dd(f,\alpha) = t\,.
 \end{align*}
 and our claim follows directly.
\end{proof}
To compute the nonlinearity of $f$ we can use Algorithm \ref{algNLoverQ} with input $f$.\\
%
\begin{algorithm}[H]
\caption{To compute the nonlinearity of the Boolean function $f$}
\label{algNLoverQ}
\begin{algorithmic}[1]
\REQUIRE{$f$}
\ENSURE{nonlinearity of $f$}
\STATE{Compute $\nP_f$}
\STATE{$j \leftarrow 1$}
\WHILE{$\V(\NLI_f^j) = \emptyset$}
  \STATE{$j \leftarrow j+1$}
\ENDWHILE
\RETURN{j} 
\end{algorithmic}
\end{algorithm}
Algorithm \ref{algNLoverQ} can be modified to eliminate the while cycle. Instead of
checking if a solution of the system
\begin{align}\label{eq:NonlinearityPolynomialSystem_j}
 \begin{cases}
  a_0^2-a_0= 0 \\
  \ldots \\
  a_n^2-a_n = 0 \\
  \nP_C(a_0,\ldots,a_n) - j = 0
 \end{cases}
\end{align}
exists in the affine algebra $\QQ/\langle a_0^2-a_0 , \ldots, a_n^2-a_n  \rangle$ for each $j \in \{1,\ldots,2^n\}$,
we can add the variable $t$ to the system
\begin{align}\label{eq:NonlinearityPolynomialSystem_t}
 \begin{cases}
  a_0^2-a_0= 0 \\
  \ldots \\
  a_n^2-a_n = 0 \\
  \nP_C(a_0,\ldots,a_n) - t = 0
 \end{cases}
\end{align}
and solve it in $\QQ[t]/\langle a_0^2-a_0 , \ldots, a_n^2-a_n  \rangle$, with respect to lexicographical monomial ordering, to find as a solution a polynomial $\tP(t)$, whose zeros are integers, representing the possible distances of the Boolean function $f$ from the affine functions. We are interested in the smallest solution of $\tP(t)$. \\
We did not investigate further which of the two solutions is best.
%
%
\section{An improvement using fast polynomial evaluation}
  \label{secNLwithFPE}
%
Once the nonlinearity polynomial $\nP_f$ is defined, we can use another approach to compute the nonlinearity avoiding the hard task of solving a polynomial system of equations.\\ 
We have to find the minimum nonnegative integer $t$ in the set of the evaluations of $\nP_f$, that is, in $\{\nP_f(\bar{a}) \mid \bar{a} \in \{0,1\}^{n+1} \subset \ZZ^{n+1}\}$.\\
We write explicitly the modified algorithm.
\begin{algorithm}[H]
\caption{To compute the nonlinearity of the Boolean function $f$}
\label{algNLfromNLP}
\begin{algorithmic}[1]
\REQUIRE{$f$}
\ENSURE{nonlinearity of $f$}
\IF{$f \in \mathcal{A}_n$}
  \RETURN{$0$}
\ELSE
  \STATE{Compute $\nP_f$}
  \STATE{Compute $m = \min\{\nP_f(\bar{a}) \mid \bar{a} \in \{0,1\}^{n+1} \}$}
  \RETURN{$m$}
\ENDIF
\end{algorithmic}
\end{algorithm}

\begin{example}
 Consider the case $n=2$, $f(x_1,x_2) = x_1x_2 + 1$. We have that $\ef = (1,1,1,0)$ and $\underline{\mathfrak{g}_n}=(a_0,a_0+a_1,a_0+a_2,a_0+a_1+a_2)$.\\ 
 Let us compute all $f_i^{(\FF)}=(\underline{\mathfrak{g}_n}+\ef)_i$ and $f_i^{(\ZZ)}$,for $i=1,\ldots,2^2$:
 \begin{align*}
  f_1^{(\FF)} & = a_0 + 1          & \rightarrow f_1^{(\ZZ)} &= -a_0 + 1\\  
  f_2^{(\FF)} & = a_0 + a_1 + 1    & \rightarrow f_2^{(\ZZ)} &= 2a_0a_1 - a_0 - a_1 + 1\\  
  f_3^{(\FF)} & = a_0 + a_2 + 1    & \rightarrow f_3^{(\ZZ)} &= 2a_0a_2 - a_0 - a_2 + 1\\  
  f_4^{(\FF)} & = a_0 + a_1  + a_2 & \rightarrow f_4^{(\ZZ)} &= 4a_0a_1a_2 - 2a_0a_1 - 2a_0a_2 \\
              &                    &                         & + a_0 - 2a_1a_2 + a_1 + a_2
 \end{align*}
 Then $\nP_f = f_1^{(\ZZ)} + f_2^{(\ZZ)} + f_3^{(\ZZ)} + f_4^{(\ZZ)} = 4a_0a_1a_2 - 2a_0 - 2a_1a_2 + 3$ and since
 $$\underline{\nP_f} = (3,1,3,1,3,1,1,3)$$
 then the nonlinearity of $f$ is $1$.\\
 Observe that the vector $\underline{\nP_f}$ represents all the distances of $f$ from all possible affine functions in $2$ variables, that is, from $0,1,x_1,x_1+1,x_2,x_2+1,x_1+x_2,x_1+x_2+1$.
\end{example}

\section{Complexity considerations}
  \label{secNLComplexity}
%
First we recall that the complexity of computing the nonlinearity of a Boolean function with $n$ variables, having as input its coefficients vector, is ${\Cal O}(n2^n)$ using the Fast M\"obius and the Fast Walsh Transform.\\
We now want to analyze the complexity of Algorithm \ref{algNLoverF2}, \ref{algNLoverQ}, \ref{algNLfromNLP}.\\
%
%
The complexity of constructing the nonlinearity polynomial is claimed in an unpublished preprint as follows:
\begin{theorem}[\cite{CGC-cry-art-BellSimSala14}]
\label{thmNLPn2n}
 There exists an algorithm to compute the nonlinearity polynomial, which requires:
 \begin{enumerate}
  \item ${\Cal O}(n2^n)$ integer sums and doublings.\\
        In particular $n2^{n}$ integer sums and $n2^{n-1}$ integer doublings, i.e. the big ${\Cal O}$ constant is $c=3/2$, provided doubling costs as summing.
  \item the storage of ${\Cal O}(2^n)$ integers of size less than or equal to $2^n$.
 \end{enumerate}
\end{theorem}
\subsection{Some considerations on Algorithm \ref{algNLoverF2}}
In Algorithm \ref{algNLoverF2}, almost all the computations are wasted evaluating all possible simple-$t$-monomials in $2^n$ variables, which are $\binom{2^n}{t}$. This number grows enormously even for small values of $n$ and $t$. We investigated experimentally how many of the $\binom{2^n}{t}$ monomials are actually needed to compute the final \Gr\ basis of $J_t^n$. Our experiment ran over all possible Boolean functions in 3 and 4 variables. The results are reported in Tables \ref{tabNumberOfMonomialsJt3},
\ref{tabNumberOfMonomialsJt4t1_3} and \ref{tabNumberOfMonomialsJt4t4_7}.\\
In this tables, for each $J_t^n$ there are four columns. Let $G_t^n$ be the \Gr\ basis of $J_t^n$. \\
Under the column labeled $\#$C we report the average number of \emph{checked} monomials in $2^n$ variables before obtaining $G_t^n$. \\
Under the column labeled $\#$S we report the average number of monomials which are actually \emph{sufficient} to obtain $G_t^n$. \\
Under the columns labeled ``m'' e ``M'' we report, respectively, the minimum and the maximum number of sufficient monomials to find $G_t^n$ running through all possible Boolean functions in $n$ variables.
\\
For example, to compute the \Gr\ basis of the ideal $J_2^3$ associated to a Boolean function $f$ whose nonlinearity is $2$, we needed to check on average 24 monomials before finding the correct basis. Between the $24$ monomials only $9.7$ (on average) were sufficient to obtain the same basis, where the number of sufficient monomials never exceeded the range $8-11$.
\begin{table*}[ht]
\begin{center}
\resizebox{14cm}{!} {
\begin{tabular}{c|*3{*4{c}|}}
    & \multicolumn{4}{c}{$J_1^3$} &   \multicolumn{4}{c}{$J_2^3$} &   \multicolumn{4}{c}{$J_3^3$} \\
 NL & $\#$S & m & M & $\#$C &     $\#$S & m & M & $\#$C &     $\#$S & m & M & $\#$C \\
\hline
  0 &   4   & 4 & 4 & 8    & 0   & 0 &  0 & 0     & 0   & 0 & 0  & 0  \\
  1 &   4.5 & 4 & 5 & 4.4  & 8.5 & 7 & 10 & 28    & 0   & 0 & 0  & 0  \\
  2 &   4.4 & 4 & 5 & 4    & 9.7 & 8 & 11 & 24    & 9.3 & 8 & 11 & 56
%
\end{tabular}
}
\end{center}
\caption{Number of monomials needed to compute the \Gr\ basis of the ideal $J_t^3$.}
\label{tabNumberOfMonomialsJt3}
\end{table*}
\begin{table*}[ht]
\begin{center}
\resizebox{14cm}{!} {
\begin{tabular}{c|*3{*4{l}|}}
 & \multicolumn{4}{c}{$J_1^4$} &   \multicolumn{4}{c}{$J_2^4$} &   \multicolumn{4}{c}{$J_3^4$} \\
NL & $\#$S & m & M & $\#$C &     $\#$S & m & M & $\#$C &     $\#$S & m & M & $\#$C \\
\hline
  0 &   5    & 5 & 5 & 16     &  0    & 0 & 0  & 0       & 0 & 0 & 0 & 0              \\  
  1 &   5.25 & 4 & 6 & 8      &  8.75 & 8 & 11 & 120     & 0 & 0 & 0 & 0              \\
  2 &   4.83 & 4 & 6 & 5.67   &  9.97 & 8 & 12 & 62.83   & 14.50 & 12 & 18 & 560      \\
  3 &   4.62 & 4 & 6 & 4.76   &  9.92 & 8 & 12 & 42.72   & 15.76 & 13 & 19 & 315.04   \\
  4 &   4.53 & 4 & 6 & 4.42   &  9.83 & 8 & 12 & 37.49   & 15.81 & 13 & 19 & 246.19   \\
  5 &   4.46 & 4 & 5 & 4.19   & 10.11 & 8 & 12 & 34.39   & 15.89 & 13 & 19 & 215.68   \\
  6 &   4.43 & 4 & 5 & 4.00   &  9.71 & 8 & 11 & 24.00   & 17.29 & 16 & 19 & 156.86   
\end{tabular}
}
\end{center}
\caption{Number of monomials needed to compute the \Gr\ basis of the ideal $J_t^4$, $t=1,2,3$.}
\label{tabNumberOfMonomialsJt4t1_3}
\end{table*}
\begin{table*}[ht]
\begin{center}
\resizebox{14cm}{!} {
\begin{tabular}{c|*4{*4{l}|}}
 & \multicolumn{4}{c}{$J_4^4$} &   \multicolumn{4}{c}{$J_5^4$} &   \multicolumn{4}{c}{$J_6^4$} &   \multicolumn{4}{c}{$J_7^4$}    \\
NL & $\#$S & m & M & $\#$C &     $\#$S & m & M & $\#$C &     $\#$S & m & M & $\#$C &     $\#$S & m & M & $\#$C \\
\hline
  0   & 0 & 0 & 0 & 0               & 0 & 0 & 0 & 0               & 0 & 0 & 0 & 0               & 0 & 0 & 0 & 0  \\  
  1   & 0 & 0 & 0 & 0               & 0 & 0 & 0 & 0               & 0 & 0 & 0 & 0               & 0 & 0 & 0 & 0  \\
  2   & 0 & 0 & 0 & 0               & 0 & 0 & 0 & 0               & 0 & 0 & 0 & 0               & 0 & 0 & 0 & 0  \\
  3   & 20.18 & 15 & 23 & 1820      & 0 & 0 & 0 & 0               & 0 & 0 & 0 & 0               & 0 & 0 & 0 & 0  \\
  4   & 21.44 & 16 & 24 & 1319.96   & 23.99 & 22 & 29 & 4368      & 0 & 0 & 0 & 0               & 0 & 0 & 0 & 0  \\
  5   & 21.54 & 19 & 24 & 1003.15   & 26.00 & 24 & 28 & 3851.24   & 23.50 & 22 & 25 & 8008      & 0 & 0 & 0 & 0  \\
  6   & 19.57 & 19 & 20 &  671.71   & 28    & 28 & 28 & 2603.79   & 28 & 28 & 28 & 7608.79      & 16 & 16 & 16 & 11441
\end{tabular}
}
\end{center}
\caption{Number of monomials needed to compute the \Gr\ basis of the ideal $J_t^4$,$t=4,5,6,7$.}
\label{tabNumberOfMonomialsJt4t4_7}
\end{table*}
\subsection{Algorithm \ref{algNLoverF2} and \ref{algNLoverQ}}
Since the ideal $J_t^n(f)$ of Definition \ref{defIdealF2} is derived from the evaluation of $\binom{2^n}{t}$ monomials (generating at most the same number of equations), then the complexity of Algorithm \ref{algNLoverF2} is equivalent to the complexity of 
solving a polynomial system
of at most $\binom{2^n}{t}$ equations of degree $d$ (where $1 < d \le t$) in $n+1$ variables over the field $\FF$. This method becomes almost impractical for $n=5$.
We recall that $t\le 2^{n-1}-2^{\frac{n}{2}-1}$ (see Equation \ref{eqMaxNL}).\\
\indent
The complexity of Algorithm \ref{algNLoverQ} is equivalent to the complexity of 
solving a polynomial system
of only  $n+1$ field equations plus one single polynomial $\nP_f$ of degree at most $n+1$ in $n+1$ variables over the field $\QQ$ (or over a prime field $\FF_p$) with coefficients of size less then or equal to $2^n$. \\
Solving the system by computing its \Gr\ basis over a prime field $\FF_p$ with $p\sim 2^n$ is much faster than computing the same base over $\QQ$. It may be investigated if there are better size for the prime $p$, or even faster specialized algorithms to solve the system.
%
%
%
%
\section{Acknowledgments}
  \label{secAck}
These results appear partially in the first author's PHD thesis and so he would like to thank the second author and the third author (his supervisor).\\
We presented our algorithms at a computation presentation in the conference MEGA \cite{CGC-cry-talk-bellini2015}.


\bibliography{RefsCGC}


\end{document}